\documentclass[11pt]{article}%
\usepackage{amssymb}
\usepackage{amsfonts}
\usepackage{amsmath}
\usepackage{amsbsy}
\usepackage{graphicx}
\usepackage{enumerate}
\usepackage{geometry}
\usepackage{amssymb}
\usepackage{amsfonts}
\usepackage{amsmath}
\usepackage{amsbsy}
\usepackage{graphicx}
\usepackage{enumerate}
\usepackage{shuffle}
\usepackage{proof}

\usepackage{tikz}
\usetikzlibrary{arrows,shapes.geometric,positioning}

\usepackage[hidelinks]{hyperref}

\providecommand{\U}[1]{\protect\rule{.1in}{.1in}}

\allowdisplaybreaks
\newtheorem{theorem}{Theorem}

\newtheorem{corollary}[theorem]{Corollary}

\newtheorem{definition}[theorem]{Definition}
\newtheorem{example}[theorem]{Example}

\newtheorem{proposition}[theorem]{Proposition}
\newtheorem{remark}[theorem]{Remark}

\newenvironment{proof}[1][Proof]{\noindent\textbf{#1.} }{\ \rule{0.5em}{0.5em}}
\definecolor{darkorange}{rgb}{1.0, 0.55, 0.0}

\allowdisplaybreaks
\newcommand{\B}{{\mathcal{B}}}

\newcommand{\act}{{\mathrm{active}}}

\newcommand{\hide}[1]{}
\providecommand{\keywords}[1]{\textbf{Keywords:} #1}

\begin{document}

\title{Dynamic reconfiguration of component-based systems described by propositional configuration logic}
	\author{George Rahonis$^{a,}$\ and Melpomeni Soula$^{b}$ \\Department of Mathematics\\Aristotle University of Thessaloniki\\54124 Thessaloniki, Greece\\$^{a}$grahonis@math.auth.gr, $^{b}$soulmelp@math.auth.gr}
	\date{}
	\maketitle


\begin{abstract}
We investigate dynamic reconfigurable component-based systems whose architectures are described by formulas of Propositional Configuration Logics. We present several examples of reconfigurable systems based on well-known architectures, and state preliminary decidability results.

\keywords{Dynamic reconfigurable component-based systems \and Architectures \and Propositional Configuration Logic.}
\end{abstract}

\section{Introduction}

One of the most rigorous formalisms in systems engineering is the so-called component-based. Component-based formalism provides a strict mathematical representation of complex systems and supports specification languages (cf. for instance \cite{Ma:Co,Pi:Ar}) as well as a strict description of important characteristics such as (dynamic) reconfiguration \cite{Ho:Te}. A component-based system consists of a finite number of different components types (usually represented by transition labelled systems) with additional several copies of them called instances. Communication among component instances is achieved by their interfaces, i.e., ports. Specifically, ports of different components connect to each other via interactions. Sets of interactions obey to specific rules which determine the topology and operation of component-systems. The sets of rules are  described by architectures. There are several ways to represent  architectures; logics were proved very reliable and moreover conclude important results (cf. for instance \cite{Ma:Co, Pa:We, Pi:Ar}). On the other hand, several practical applications, like IoT systems, are described with component-based systems where the number of instances of several of their components changes dynamically, i.e., some new instances can be connected, some other may be inactivated, etc. Though in the dynamic reconfiguration process, fault interactions may occur and this in turn  affects the right implementation  of the underlying  architecture of the system. It is a challenge to decide whether a dynamic reconfiguration of a component-based system preserves its main characteristics and if not, whether there is a possibility to "correct" it. \\
In this paper, a work in progress, we suggest a method based on propositional configuration logics (PCL for short). PCL which was introduced in \cite{Ma:Co}, provides a strict mathematical specification language to describe architectures of component-based systems. Almost every known architecture can be represented by a formula in PCL and its first- and second-order levels. PCL has been also extended to a quantitative set-up (cf. \cite{Ka:We,Pa:We}). In our model, the reconfiguration of a component-based and thus of  the implementation of its architecture is described by  an assignment mapping $\sigma$ which activates/inactivates  interactions among the component instances. We show that we can decide whether such an infinite assignment ensures that the architecture is finally (i.e., after a "time" to infinity) rightly implemented. If this is not the case, we investigate whether a restriction of $\sigma$ provides a "right" assignment.

\section{Related work}
There are several approaches to describe reconfiguration of component-based systems. In \cite{Al:Sp,Ba:Pr,Ma:Re,Me:AD,Se:Sp} the authors study the problem of dynamic reconfigurable architectures by building several software tools, where as in \cite{Ma:Dy} and \cite{Ca:Su} architectures' reconfiguration is studied with $\pi$-calculus and $\pi$-ADL, respectively. Several types of graph grammars models were used also for the investigation of dynamic changing of architectures \cite{Br:Mo,Bu:Dy,Ha:St}. In \cite{Bo:Lo} a resource logic is used as a tool for supporting local reasoning of parametric and reconfigurable architectures. A tool for service-oriented reconfigurable architectures, based on LTL, was built in \cite{Fi:Mo}.

Our work is more closely to the recent paper \cite{Ho:Te}. In that paper the authors introduced a temporal configuration logic, i.e., an LTL whose atomic  propositions are configuration logic formulas, in order to describe  architectures of reconfigurable component-based systems.  The unknown number of component instances is represented by the first- and second-order levels of PCL. In addition, with those two levels the authors could describe architectures such as the ones with the ring structure which, as it is well known, cannot be described by formulas in PCL. 
In our theory we assume that we have a bounded number of instances of every component type. The "bound" can be as large as we want, for instance all the components instances that should exist in the real word. Then, these numbers can be trivially expanded in formulas of our logic. We show that most of the well-known architectures can described by formulas in our logic, though architectures requiring first-and second-order levels of configuration logic like ring, grid, etc. remain out of our study for the time being. The advantage of our logic lies on the produced decidability results.

	\section{Preliminaries}
	For every natural number $n>0$, we set $[n]=\{1, \ldots , n\}$. Let $A$ be an alphabet, i.e, a non-empty finite set. As usually the set of all infinite words over A is denoted by $A^\omega$. A (non-deterministic) \emph{Büchi automaton} is a five-tuple $\mathcal{A}=(Q,A,I,\Delta,F)$, where $Q$ is a finite set of states, $A$ is the input alphabet, $I\subseteq Q$ is the set of initial states, $\Delta \subseteq Q\times A \times Q$ is the set of transitions and $F\subseteq Q$ is the set of final states. 
		
	Let $w=a_0a_1...\in A^\omega.$ A path of $\mathcal{A}$ over $w$ is an infinite sequence of transitions  $P_{w}= \left(\left(q_i,a_i,q_{a_{i+1}} \right)  \right)_{i\geq 0},$ where $ (q_i,a_i,q_{a_{i+1}} )\in \Delta$ for every $i\geq0.$  We denote by $In^{Q}(P_{w})$ the set of states occurring infinitely many times along $P_w$. Then, $P_w$ is \emph{successful} if $In^{Q}(P_{w}) \cap F\neq \emptyset$. An infinite word $w\in A^\omega$ is \emph{recognized} (or \emph{accepted}) by $\mathcal{A}$ if there is a successful path of $\mathcal{A}$ over $w.$ The \emph{language} (or \emph{behavior}) of  $\mathcal{A}$ which is denoted by $L(\mathcal{A})$ is the set of all words accepted by $\mathcal{A}$.
	
Next, we recall propositional interaction and configuration logics (cf. \cite{Ma:Co}). For this, we consider a  non-empty finite set $P$ whose elements are called ports. We let $I(P)=\mathcal{P}
	(P)\setminus\{\emptyset\}$, where $\mathcal{P}(P)$ denotes the power set of $P$. The elements of $I(P)$ are called interactions. Then, the syntax of propositional interaction logic
	(PIL for short) formulas over $P$\ is given by the grammar
	$
	\phi::=\mathrm{true}\mid p\mid\overline{\phi}\mid\phi\vee\phi
	$ 
	where $p\in P$. We set $\overline{\overline{\phi}}=\phi$ for every
	PIL formula $\phi$ and $\mathrm{false}=\overline{\mathrm{true}}$. The conjunction
	of two PIL formulas $\phi,\phi^{\prime}$ is defined by $\phi\wedge
	\phi^{\prime}:=\overline{\left(  \overline{\phi}\vee\overline{\phi^{\prime}%
		}\right)  }$. A PIL formula of the form $p_{1}\wedge\ldots\wedge p_{n}$
	with $n>0$, and $p_{i}\in P$ or $p_i=\overline{p'_{i}}$ with $p'_i\in P$ for every $1\leq i\leq
	n$, is called a \emph{monomial}. We denote a monomial
	$p_{1}\wedge\ldots\wedge p_{n}$ by $p_{1}\ldots p_{n}$. 	
	Let $\phi$ be a PIL formula and $a \in I(P)$ an interaction. We define the satisfaction relation 
	$a\models_{\mathrm{PIL}}\phi$, by induction on the structure of $\phi$, in the following way:
\begin{itemize}
\item[-] $a\models_{\mathrm{PIL}} \mathrm{true}$,
 
\item[-] $a\models_{\mathrm{PIL}} p$ iff  $p \in a$,

\item[-] $a\models_{\mathrm{PIL}} \overline{\phi}$  iff  $ a \not
 \models_{\mathrm{PIL}} \phi$,

\item[-] $a\models_{\mathrm{PIL}} \phi_1 \vee \phi_2$ iff  $a \models_{\mathrm{PIL}} \phi_1$ or $a \models_{\mathrm{PIL}} \phi_2.$
\end{itemize}
	
	\noindent It should be clear that
	$a\not \models _{\mathrm{PIL}} \mathrm{false}$ for every $a\in I(P)$. For every interaction $a$ we
	define its characteristic monomial $m_{a}%
	=\bigwedge\nolimits_{p\in a}p\wedge\bigwedge\nolimits_{p\notin a}\overline{p}%
	$. Then, for every $a^{\prime} \in I(P)$ we trivially get $a^{\prime
	}\models_{\mathrm{PIL}}m_{a}$ iff $a'= a$.
	
	The syntax of \emph{propositional configuration logic }(PCL for short) formulas over $P$\ is given by the grammar $
	f::=\mathrm{true}\mid\phi\mid\lnot f\mid f\sqcup f\mid f+f
$ where $\phi$ denotes a PIL formula over $P$. The operators $\lnot$, $\sqcup$, and
	$+$ are called complementation, union, and coalescing,
	respectively. The intersection $\sqcap$ operator is defined  by:  $f_{1}\sqcap f_{2}:=\lnot(\lnot f_{1}\sqcup\lnot f_{2}).$ We set
	$C(P)=\mathcal{P}(I(P))\setminus\{\emptyset\}$. For every PCL formula
	$f$ and $\gamma\in C(P)$ we define the satisfaction relation $\gamma\models
	f$\ inductively on the structure of $f$ in the following way:
	
	\begin{itemize}
	\item[-]	$\gamma\models \mathrm{true}$,
	\item[-]  $\gamma\models\phi$ iff  $a \models_{\mathrm{PIL}}\phi$ for every $a\in\gamma$, 
	\item[-] $\gamma\models\lnot f$ iff  $\gamma\not \models f$,
	\item[-] $ \gamma\models f_{1}\sqcup f_{2}$ iff $\gamma\models f_{1}$ or  $\gamma\models f_{2}$,
\item[-] $\gamma\models f_{1}+f_{2}$ iff there exist $\gamma
		_{1},\gamma_{2}\in C(P)$ such that $\gamma=\gamma_{1}\cup\gamma_{2}$, \\
		\qquad \qquad \qquad \qquad and $\gamma_{1}\models f_{1}$ and $\gamma_{2}\models f_{2}$.
	\end{itemize}
	
	\noindent Then, we trivially get
	
	- $\gamma\models f_{1}\sqcap f_{2}$ \ \ iff \ \ $\gamma\models f_{1}$
	and $\gamma\models f_{2},$ and
	
	- $\gamma\models f_{1}\implies f_{2}$ \ \ iff \ \ $\gamma
	\not \models f_{1}$ or $\gamma\models f_{2}.$
	
	\noindent We define the closure ${\sim} f$ of every
	PCL formula $f$ by
	
	- ${\sim} f:=f+\mathrm{true}$,
	
	\noindent and the disjunction $f_{1}\vee f_{2}$ of two PCL
	formulas $f_{1}$ and $f_{2}$ by
	
	- $f_{1}\vee f_{2}:=f_{1}\sqcup f_{2}\sqcup (f_{1}+f_{2})$.

	\noindent Two PCL formulas $f,f^{\prime}$ are called \emph{equivalent}, and we
		denote it by $f\equiv f^{\prime},$\ whenever $\gamma\models f$\ iff
		$\gamma\models f^{\prime}$\ for every $\gamma\in C(P)$.\\
		Since  $\phi\wedge\phi^{\prime}\equiv\phi\sqcap\phi^{\prime}$ for every interaction formula $\phi$, we denote in the sequel both conjunction
	and intersection operations of PCL formulas with the same symbol $\wedge$.
	
	 We say that a PCL formula $f$ is in \emph{full normal form} if $$f=\bigsqcup\nolimits_{i\in I}\sum\nolimits_{j\in J_{i}}m_{i,j}$$ 
where the
index sets $I$ and $J_{i}$, for every $i\in I$, are finite and $m_{i,j}$'s are
\emph{full monomials}, i.e., monomials of the form $\bigwedge_{p\in P_{+}}p\wedge
\bigwedge_{p\in P\_}\overline{p}$ with $P_{+}\cup P\_=P$ and $P_{+}\cap
P\_=\emptyset$.

	\section{Architectures}
	\label{comp_based_system}
	We introduce  component-based systems and architectures. A component-based system consists of a finite number of components of the same or different type. In most of the well-known component-based systems including BIP \cite{Bl:Al}, REO \cite{Am:RE}, X-MAN \cite{He:Co}, and B \cite{Al:Th} the components  are defined as labelled transition systems (LTS for short). Here, we use the same definition and follow the terminology of BIP framework for the basic notions. The communication among components is implemented  through their corresponding interfaces, i.e, by the associated set of labels, called ports. More precisely, communication of components is defined by interactions, i.e., sets of ports.

		An \emph{atomic component} is an \emph{LTS} $B=(Q,P,q_0,R)$ where $Q$ is a finite
		set of \emph{states}, $P$ is a finite set of \emph{ports}, $q_0$ is the \emph{initial state}, and $R\subseteq Q \times P \times Q$ is the set of \emph{transitions}.

	An atomic component $B$ will be called a \emph{component}, whenever we deal with several atomic components. Furthermore, for every set $\mathcal{B} = \{B(i) \mid   i \in [n] \} $ of components, with $B(i)=(Q(i),P(i),q_{0}(i),R(i))$, $i \in [n]$, we  assume that  $(Q(i) \cup P(i))\cap (Q(i') \cup P(i')) = \emptyset $ 
	for every $1 \leq i\neq i' \leq n$. 
	
	Let $\mathcal{B} = \{B(i) \mid i \in [n]   \} $ be a set of components. We let $P_{\mathcal{B}} =\bigcup_{i \in [n]}  P(i)$ denoting the set of all ports of the elements of $\mathcal{B}$. An \emph{interaction of} $\B$ is an  interaction $a \in I(P_{\mathcal{B}})$ such that  $\vert a \cap P(i)\vert \leq 1$, for every $i \in [n]$. If $p \in a$, then  $p$ is active in $a$. We shall denote by $I_{\mathcal{B}}$ the set of all interactions of $\mathcal{B}$, i.e., 
	$$I_{\mathcal{B}} = \left \{ a \in I(P_{\mathcal{B}}) \mid   \vert a \cap P(i)\vert \leq 1 \text{ for every } i \in [n] \right\}$$
	and  let 
	$$C_{\mathcal{B}} = \mathcal{P}(I_{\mathcal{B}})\setminus \{\emptyset\}.$$
	\begin{definition} \label{BIP-def}An \emph{architecture} is a pair $(\B, f)$ where $\B=\{B(i) \mid i \in [n]  \}$ is a set of components, with $B(i)=(Q(i),P(i),q_{0}(i),R(i))$ for every $i \in [n]$, and $f$  is a \emph{PCL} formula over  $P_{\B}$.
	\end{definition}
	
	In order to define reconfigurable architectures we need to extend the notion of component-based systems\footnote{In fact we use a slightly modified notion of parametric component-systems of \cite{Pi:Ar}}. Specifically, let $\B=\{B(i) \mid i \in [n]  \}$ be a set of component types. For every $i \in [n]$ we consider a natural number $u_i >0$ and for every $j \in [u_i]$  a copy $B(i,j)=(Q(i,j),P(i,j),q_{0}(i,j),R(i,j))$ of $B(i)$ which is called the \emph{$j$-th instance of} $B(i)$. Thus, for every $i \in [n]$ and $j \in [u_i]$, the instance $B(i,j)$ is also a component and we call it a   \emph{component instance}. We impose that $(Q(i,j) \cup P(i,j)) \cap (Q(i', j') \cup  P(i',j')) = \emptyset$ whenever $i \neq i' $ or $j \neq j'$ for every $ i, i' \in [n]$ and $j, j' \in \bigcup_{i\in [n]} [u_i] $.  This restriction is needed in order to identify the distinct component instances. It also permits us to use, without any confusion, the notation  $P(i,j)=\{p(j) \mid p\in P(i)\}$ for every $i \in [n]$ and $j \in [u_i]$.   We denote by $p\mathcal{B}$ the derived component-based system, i.e., $p\mathcal{B}= \{B(i,j) \mid i \in [n],  j \in [u_i] \} $.  The set of ports of $p\B$ is defined by $P_{p\B} =\bigcup_{i \in [n], j \in [u_i]}  P(i,j)$. The set of its interactions is given by 
	$$I_{p\mathcal{B}} = \left \{ a \in I(P_{p\mathcal{B}}) \mid   \vert a \cap P(i,j)\vert \leq 1 \text{ for every } i \in [n], j\in [u_i] \right\}$$
	and  we let 
	$$C_{p\mathcal{B}} = \mathcal{P}(I_{p\mathcal{B}})\setminus \{\emptyset\}.$$
	Next, let $p\mathcal{B}= \{B(i,j) \mid i \in [n],  j \in [u_i]\}$ be a component-based system. An \emph{implementation} $\sigma$ \emph{on} $p\mathcal{B}$ is a sequence of mappings $\sigma=\left(\sigma_l\right)_{l\geq 0}$ with $\sigma_l :I_{p\mathcal{B}} \rightarrow \{\mathrm{active, inactive}\}$. For every $k \geq 0$ and $\gamma \in C_{p\mathcal{B}}$ we set $\sigma_k(\gamma)=\mathrm{active}$ if  $\sigma_k(a)=\mathrm{active}$ for every $a \in \gamma$, otherwise we set $\sigma_k(\gamma)=\mathrm{inactive}$.

	\begin{definition}
		A \emph{dynamic reconfigurable architecture} (\emph{dra} for short) is a triple $(p\B, f, \sigma)$ where $p\mathcal{B}=\{B(i,j) \mid i \in [n],  j \in [u_i] \} $ is a component-based system, $f$ a \emph{PCL} formula over $P_{p\mathcal{B}}$, and $\sigma$ an implementation on $p\B$.  
	\end{definition}
	
	\noindent For every $l\geq 0$ we set 
	$$\gamma(\sigma_l)=\{a \in I_{p\B} \mid \sigma_l(a)=\act\}$$
	and define the satisfaction of $f$ by $\sigma_l$ as follows: 
	$$\sigma_l \models f  \text{  \ \ \  iff \ \ \   }\gamma(\sigma_l) \models f.$$ 
	Then, we call the implementation $\sigma$ \emph{trustworthy for} $f$ if there is a natural number $g\geq 0$ such that $\sigma_l \models f$ for every $l>g$. The meaning of $\sigma$ being trustworthy is that after a "time" $g$ the set of  active interactions, determined by $\sigma_l$, for every $l>g$, satisfies the formula $f$, and thus the architecture described by $f$ is always  rightly implemented.

\begin{remark}\label{rem_imp}
From the constructions above, every implementation $\sigma=\left(\sigma_l\right)_{l \geq 0}$ on $p\B$ defines a sequence $\gamma=\left( \gamma(\sigma_l)\right)_{l\geq0}$ of sets of interactions in $C_{p\B}$. Conversely, every sequence $\gamma=\left( \gamma_l\right)_{l\geq0}$ of elements in $C_{p\B}$ determines an implementation $\sigma(\gamma)=\left(\sigma(\gamma)_l\right)_{l \geq 0}$ on $p\B$ with $\sigma(\gamma)_l(a)=\mathrm{active}$ if $a \in \gamma_l$, and $\sigma(\gamma)_l(a)=\mathrm{inactive}$ otherwise, for every $a \in I_{p\B}$. Therefore, there is a one-to-one correspondence of implementations on $p\B$ and sequences of sets of  interactions in $C_{p\B}$. 
\end{remark}	

\subsection{Examples}
	In the sequel, we present PCL formulas of several well-known architectures for component-based systems with arbitrarily many components.   
	
	\begin{example}
		We present  a \emph{PCL} formula for the \emph{Master-Slave}  architecture considered as dynamic reconfigurable. This architecture is composed by two component types, namely \emph{masters} denoted by $B(1)$ and \emph{slaves} denoted by $B(2)$. Both of them have one port each, denoted respectively by $m$ and $s$. Connections are permitted only among masters and slaves via interactions of type
		$\{m,s\}$. Moreover, every slave instance can be connected to a unique master instance. An
		instantiation of the architecture for two masters and two slaves with all the possible cases of
		the allowed interactions if all slaves are connected is shown in Figure \ref{m-s}.
		\begin{figure}[h]
			\centering
			\resizebox{1.0\linewidth}{!}{
				\begin{tikzpicture}[>=stealth',shorten >=1pt,auto,node distance=1cm,baseline=(current bounding box.north)]
					\tikzstyle{component}=[rectangle,ultra thin,draw=black!75,align=center,inner sep=9pt,minimum size=2.5cm,minimum height=1.8cm, minimum width=2.7cm]
					\tikzstyle{port}=[rectangle,ultra thin,draw=black!75,minimum size=7.5mm]
					\tikzstyle{bubble} = [fill,shape=circle,minimum size=5pt,inner sep=0pt]
					\tikzstyle{type} = [draw=none,fill=none]

					\node [component,align=center] (a1)  {};
					\node [port] (a2) [below=-0.76cm of a1]  {
						$m$};
					\node[bubble] (a3) [below=-0.105cm of a1]   {};

					\node[type]  [above=-0.6cm of a1]{{\small Master $B(1,1)$}};
					\node [component] (a4) [below=2cm of a1]  {};
					\node [port,align=center,inner sep=5pt] (a5) [above=-0.75cm of a4]  {$s$};
					\node[bubble] (a6) [above=-0.105cm of a4]   {};
						\node []  (w2)  [below=0.5cm of a6]  {};
				
					\node[type]  [below=-0.6cm of a4]{{\small Slave $B(2,1)$}};
					\path[-]          (a1)  edge                  node {} (a4);
					\node [component] (b1) [right=1cm of a1] {};
					\node [port] (b2) [below=-0.76cm of b1]  {$m$};
					\node[bubble] (b3) [below=-0.105cm of b1]   {};
				
					\node[type]  [above=-0.6cm of b1]{{\small Master $B(1,2)$}};

					\node [component] (b4) [below=2cm of b1]  {};
					\node [port,align=center,inner sep=5pt] (b5) [above=-0.75cm of b4]  {$s$};
					\node[bubble] (b6) [above=-0.105cm of b4]   {};
					\node[type]  [below=-0.6cm of b4]{{\small Slave $B(2,2)$}};
					
					\path[-]          (b1)  edge                  node {} (b4);

					\node [component] (c1)[right=1cm of b1] {};
					\node [port] (c2) [below=-0.76cm of c1]  {$m$};
					\node[bubble] (c3) [below=-0.105cm of c1]   {};
					
					\node[type]  [above=-0.6cm of c1]{{\small Master $B(1,1)$}};

					\node [component] (c4) [below=2cm of c1]  {};
					\node [port,align=center,inner sep=5pt] (c5) [above=-0.75cm of c4]  {$s$};
					\node[bubble] (c6) [above=-0.105cm of c4]   {};
			
					\node[type]  [below=-0.6cm of c4]{{\small Slave $B(2,1)$}};
					
					\path[-]          (c1)  edge                  node {} (c4);

					\node [component] (d1)[right=1cm of c1] {};
					\node [port] (d2) [below=-0.76cm of d1]  {$m$};
					\node[bubble] (d3) [below=-0.105cm of d1]   {};
				
					\node[type]  [above=-0.6cm of d1]{{\small Master $B(1,2)$}};

					\node [component] (e4) [below=2cm of d1]  {};
					\node [port,align=center,inner sep=5pt] (e5) [above=-0.75cm of e4]  {$s$};
					\node[] (i1) [above right=-0.25 cm and -0.25cm of e5]   {};
					\node[bubble] (e6) [above=-0.105cm of e4]   {};
				
					\node[type]  [below=-0.6cm of e4]{{\small Slave $B(2,2)$}};
					
					\path[-]          (c3)  edge  node {}           (i1);

					\node [component] (f1)[right=1cm of d1] {};
					\node [port] (f2) [below=-0.76cm of f1]  {$m$};
					\node[bubble] (f3) [below=-0.105cm of f1]   {};
				
					\node[type]  [above=-0.6cm of f1]{{\small Master $B(1,1)$}};

					\node [component] (g4) [below=2cm of f1]  {};
					\node [port,align=center,inner sep=5pt] (g5) [above=-0.75cm of g4]  {$s$};
					\node[] (i2) [above left=-0.25 cm and -0.25cm of g5]   {};
					\node[bubble] (g6) [above=-0.105cm of g4]   {};
					 
					\node[type]  [below=-0.6cm of g4]{{\small Slave $B(2,1)$}};

					\node [component] (h1)[right=1cm of f1] {};
					\node [port] (h2) [below=-0.76cm of h1]  {$m$};
					\node[bubble] (h3) [below=-0.105cm of h1]   {}; 
					\node[type]  [above=-0.6cm of h1]{{\small Master $B(1,2)$}};

					\node [component] (j4) [below=2cm of h1]  {};
					\node [port,align=center,inner sep=5pt] (j5) [above=-0.75cm of j4]  {$s$};
					\node[] (i3) [above right=-0.25 cm and -0.25cm of j5]   {};
					\node[bubble] (j6) [above=-0.105cm of j4]   {}; 
					\node[type]  [below=-0.6cm of j4]{{\small Slave $B(2,2)$}};
					
					\path[-]          (h3)  edge                  node {} (i2);
					
					\path[-]          (f3)  edge                  node {} (i3);

					\node [component] (k1)[right=1cm of h1] {};
					\node [port] (k2) [below=-0.76cm of k1]  {$m$};
					\node[bubble] (k3) [below=-0.105cm of k1]   {};
					\node[type]  [above=-0.6cm of k1]{{\small Master $B(1,1)$}};

					\node [component] (k4) [below=2cm of k1]  {};
					\node [port,align=center,inner sep=5pt] (k5) [above=-0.75cm of k4]  {$s$};
					\node[] (i4) [above left=-0.25 cm and -0.25cm of k5]   {};
					\node[bubble] (k6) [above=-0.105cm of k4]   {};
					\node[type]  [below=-0.6cm of k4]{{\small Slave $B(2,1)$}};

					\node [component] (l1)[right=1cm of k1] {};
					\node [port] (l2) [below=-0.76cm of l1]  {$m$};
					\node[bubble] (l3) [below=-0.105cm of l1]   {};; 
					\node[type]  [above=-0.6cm of l1]{{\small Master $B(1,2)$}};

					\node [component] (m4) [below=2cm of l1]  {};
					\node [port,align=center,inner sep=5pt] (m5) [above=-0.75cm of m4]  {$s$};
					\node[] (i5) [above right=-0.25 cm and -0.45cm of m5]   {};
					\node[bubble] (m6) [above=-0.105cm of m4]   {}; 
					\node[type]  [below=-0.6cm of m4]{{\small Slave $B(2,2)$}};
					
					\path[-]          (l3)  edge                  node {} (i4);
					
					\path[-]          (l3)  edge                  node {} (i5);

			\end{tikzpicture}}
			\caption{ Master-Slave architecture.}
			\label{m-s}
		\end{figure}
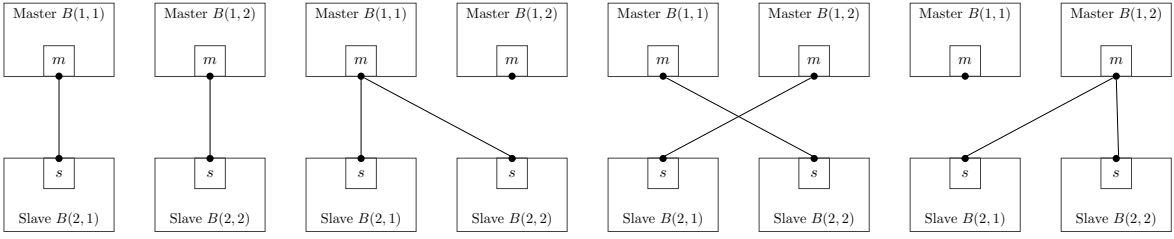

		The \emph{PIL} formula for the interaction among a master instance $B(1,k)$ and a slave instance $B(2,j)$ is given by
		$$\phi_{ms}(m(k),s(j))=m_{\{m(k),s(j)\}}.$$
		We consider the \emph{PCL} formula
		$$f=\bigvee_{j\in [u_2]}\left(\bigsqcup_{k\in[u_1]}\phi_{ms}(m(k),s(j))\right).$$
		Assume that $\gamma \in C_{p\B}$ satisfies $f$. This implies that there is at least one slave instance $B(2,j)$ ($j \in [u_2]$), which is connected to a master instance $B(1,k)$ ($k \in [u_1]$), hence the architecture is active. \\
		
	\end{example}

	\begin{example}
		We develop a \emph{PCL} formula for the well-known \emph{Publish-Subscribe} architecture (cf. \cite{Ba:Dy,Ba:On,Eg:Pu,Ol:AP,Pa:Pu,Ya:Pr,Zh:To}). For this, we define three types of components namely \emph{publishers} denoted by $B(1)$, \emph{topics} denoted by $B(2)$ and \emph{subscribers} denoted by $B(3)$. Publishers have one port namely $p$, topics have two ports namely $t_p$ and $t_s$, and subscribers one port namely $s$. Publishers are connected to topics via interactions of type $\{p,t_p\}$, and subscribers are connected to topics via interactions of type $\{s,t_s\}$.\\ 	
		 The \emph{PIL} formula for the interaction between a publisher instance $B(1,k)$ ($k \in [u_1]$) and a topic instance $B(2,j)$ ($j \in [u_2]$) is given by the \emph{PIL} formula 
		$$\phi_{pt_p}(p(k),t_p(j))= m_{\{p(k), t_p(j)\}}.$$ 
		The interaction among a subscriber instance $B(3,r)$ ($r \in [u_3])$ and a topic instance $B(2,j)$ ($j \in [u_2]$) is given by the \emph{PIL} formula 
		$$\phi_{st_s}(s(r),t_s(j))= m_{\{s(r), t_s(j)\}}.$$   
		Next, we let 
		$$\zeta(t_p(j))= \bigvee_{k \in [u_1]}\phi_{pt_p}(p(k),t_p(j))$$
		and
		$$\zeta(s(r))=\bigvee_{j \in [u_2]}(\phi_{st_s}(s(r),t_s(j))+\zeta(t_p(j))).$$
		We conclude to the \emph{PCL} formula
		$$f=\bigvee_{r \in [u_3]}\zeta(s(r))\vee \bigvee_{j\in [u_2]}\zeta (t_p(j)).$$
		One can easily see that if there is at least one subscriber instance $B(3,r)$ ($r \in [u_3]$)  connected to a topic instance $B(2,j)$ ($j \in [u_2]$)  which in turn is connected to a publisher instance $B(1,k)$ ($k \in [u_1]$), or there is at least one topic instance $B(2,j)$ ($j \in [u_2]$) connected to a publisher instance $B(1,k)$ ($k \in [u_1]$), then the corresponding set $\gamma \in C_{p\B}$ of interactions satisfies the \emph{PCL} formula $f$. A snapshot of the dynamic reconfigurable system with $u_1$ publishers, $u_2$ topics and $u_3$ subscribers, where the active interactions satisfy the formula $f$ is shown in Figure \ref{pub-sub}. \\
		
		  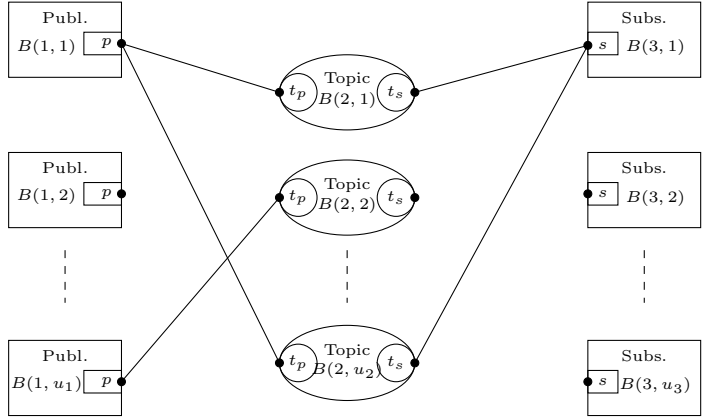
\begin{figure} 
			\begin{center}
				\begin{tikzpicture}
					
					\draw  (-0.5,0) rectangle  (1,1) ;
					\node at (0.25,0.8) {\tiny Publ.};
					\node at (0.005,0.4) {\tiny $B(1,1)$};
					\draw  (0.5,0.3) rectangle (1,0.6);
					\node at (0.8,0.45) {\tiny $p$};
					\draw[fill] (1,0.45) circle [radius=1.5pt];

					\draw  (-0.5,-2) rectangle  (1,-1) ;
					\node at (0.25,-1.2) {\tiny Publ.};
					\node at (0.005,-1.58) {\tiny $B(1,2)$};
					\draw  (0.5,-1.7) rectangle (1,-1.4);
					\node at (0.8,-1.55) {\tiny $p$};
					\draw[fill] (1,-1.55) circle [radius=1.5pt];

					\draw[dashed] (0.25,-2.3)--(0.25,-3);
					
					\draw  (-0.5,-3.5) rectangle  (1,-4.5) ;
					\node at (0.25,-3.7) {\tiny Publ.};
					\node at (0.005,-4.08) {\tiny $B(1,u_1)$};
					\draw  (0.5,-4.2) rectangle (1,-3.9);
					\node at (0.8,-4.05) {\tiny $p$};
					\draw[fill] (1,-4.05) circle [radius=1.5pt];

					\draw (4,-0.2) ellipse (0.9cm and 0.5cm);
					\node at (4,-0.02) {\tiny Topic};
					\node at (4,-0.3) {\tiny $B(2,1)$};
					\draw (3.35,-0.2) ellipse (0.25cm and 0.25cm);
					\node at (3.35,-0.2) {\tiny $t_p$};
					\draw (4.65,-0.2) ellipse (0.25cm and 0.25cm);
					\node at (4.65,-0.2) {\tiny $t_s$};
					\draw[fill] (3.1,-0.2) circle [radius=1.5pt];
					\draw[fill] (4.9,-0.2) circle [radius=1.5pt];

					\draw (4,-1.6) ellipse (0.9cm and 0.5cm);
					\node at (4,-1.44) {\tiny Topic};
					\node at (4,-1.7) {\tiny $B(2,2)$};
					\draw (3.35,-1.6) ellipse (0.25cm and 0.25cm);
					\node at (3.35,-1.6) {\tiny $t_p$};
					\draw (4.65,-1.6) ellipse (0.25cm and 0.25cm);
					\node at (4.65,-1.6) {\tiny $t_s$};
					\draw[fill] (3.1,-1.6) circle [radius=1.5pt];
					\draw[fill] (4.9,-1.6) circle [radius=1.5pt];

					\draw[dashed] (4,-2.3)--(4,-3);
					
					\draw (4,-3.8) ellipse (0.9cm and 0.5cm);
					\node at (4,-3.64) { \tiny Topic};
					\node at (4,-3.9) { \tiny $B(2,u_2)$};
					\draw (3.35,-3.8) ellipse (0.25cm and 0.25cm);
					\node at (3.35,-3.8) {\tiny $t_p $};
					\draw (4.65,-3.8) ellipse (0.25cm and 0.25cm);
					\node at (4.65,-3.8) {\tiny $t_s$};
					\draw[fill] (3.1,-3.8) circle [radius=1.5pt];
					\draw[fill] (4.9,-3.8) circle [radius=1.5pt];

					\draw  (7.2,0) rectangle  (8.7,1) ;
					\node at (7.95,0.8) {\tiny Subs.};
					\node at (8.1,0.4) {\tiny $B(3,1)$};
					\draw  (7.2,0.3) rectangle (7.6,0.6);
					\node at (7.4,0.425) {\tiny $s$};
					\draw[fill] (7.2,0.425) circle [radius=1.5pt];
					
					\draw  (7.2,-2) rectangle  (8.7,-1);
					\node at (7.95,-1.2) {\tiny Subs.};
					\node at (8.1,-1.6) {\tiny $B(3,2)$};
					\draw  (7.2,-1.7) rectangle (7.6,-1.4);
					\node at (7.4,-1.55) {\tiny $s$};
					\draw[fill] (7.2,-1.55) circle [radius=1.5pt];

					\draw[dashed] (7.95,-2.3)--(7.95,-3);

					\draw  (7.2,-3.5) rectangle  (8.7,-4.5) ;
					\node at (7.95,-3.7) {\tiny Subs.};
					\node at (8.1,-4.1) {\tiny $B(3,u_3)$};
					\draw  (7.2,-4.2) rectangle (7.6,-3.9);
					\node at (7.4,-4.05) {\tiny $s$};
					\draw[fill] (7.2,-4.05) circle [radius=1.5pt];

					\draw (1,0.45)--(3.1,-0.2);
					\draw  (1,0.45)--(3.1,-3.8);
					
					\draw (4.9,-0.2)--(7.2,0.425);

					\draw (4.9,-3.8)--(7.2,0.425);			
					
					\draw (1,-4.05)--(3.1,-1.6);	
				\end{tikzpicture}  
				
			\end{center}
			\caption{ Publish-Subscribe architecture.}
			\label{pub-sub}
		\end{figure}
	\end{example}
	\begin{example}
		
We investigate the \emph{Pipes-Filters} architecture. It involves two types of components, the \emph{Pipes}, which serve as connectors for the stream of data being transformed, and the \emph{Filters}, which perform transformations on data and process the input they receive, denoted respectively, by $B(1)$ and $B(2)$. Both of them have two ports namely $in_p, out_p$ and $in_f, out_f$,  respectively. Every filter is connected with two seperate unique pipes via ineractions of type $\{in_f,out_p\}$ and $\{out_f,in_p\}$. Also, every pipe  can be connected to at most one filter via an interaction $\{in_f,out_p\}$ and there are at least two distinguished connected pipe instances that are the ends of the pipeline, i.e., the $in_p$ of the first component  and the $out_p$ of the last component are not connected with other instances.\\
		The interaction of type $\{in_f,out_p\}$ among  a filter instance $B(2,j)$ ($j \in [u_2])$ and a pipe instance $B(1,k)$ ($k \in [u_1]$) is given by the \emph{PIL} formula 
		$$\phi_{in_fout_p}(in_f(j),out_p(k))= m_{\{in_f(j), out_p(k)\}}.$$   
		Respectively, the interaction of type $\{out_f,in_p\}$ among  a filter instance $B(2,j)$ ($j \in [u_2])$ and a pipe instance $B(1,k)$ ($k \in [u_1]$) is given by the \emph{PIL} formula 
		$$\phi_{out_fin_p}(out_f(j),in_p(k))= m_{\{out_f(j), in_p(k)\}}.$$   
		Now, we will built the \emph{PCL} formula $f$ that describes the Pipes-Filters architecture. 
		At first, we consider the \emph{PCL} formula
		$$\zeta_1(in_f(j),out_f(j))=  \bigsqcup_{i_1\in[u_1]\atop i_2\in[u_1]\setminus\{i_1\}} \left(\phi_{in_fout_p}(in_f(j),out_p(i_1))+ \phi_{out_fin_p}(out_f(j),in_p(i_2)) \wedge\xi \right)  \text{,} $$\\
		where
		$$\xi= \bigwedge_{i_1'\in [u_1]\setminus\{i_1\}} \neg {\phi_{in_fout_p}\left( in_f(j),out_p(i_1')\right) }\wedge\bigwedge_{i_2'\in [u_1]\setminus\{i_2\}}\neg {\phi_{out_fin_p}\left( out_f(j),in_p(i_2')\right) }.$$\\
		If $j\in[u_2]$, then the formula $\zeta_1(in_f(j),out_f(j))$ express that there are some dinstict indices $i_1,i_2 \in[u_1]$ such that the filter instance $B(2,j)$ is connected with the pipe instances $B(1,i_1)$,$B(1,i_2)$ via an interaction $\{in_f,out_p\}$ and $\{out_f,in_p\}$ respectively, besides the subformula $\xi$ refers that these two pipes are unique.
	 	\begin{figure}[h]

			\centering
			\resizebox{1.0\linewidth}{!}{
				\begin{tikzpicture}[>=stealth',shorten >=1pt,auto,node distance=1cm,baseline=(current bounding box.north)]
					\tikzstyle{component}=[rectangle,ultra thin,draw=black!75,align=center,inner sep=9pt,minimum size=1.5cm,minimum width=2.5cm,minimum height=2cm]
					\tikzstyle{port}=[rectangle,ultra thin,draw=black!75,minimum size=6mm]
					\tikzstyle{bubble} = [fill,shape=circle,minimum size=5pt,inner sep=0pt]
					\tikzstyle{type} = [draw=none,fill=none]

					\node [component] (a1) {};
					
					\node[bubble] (a2) [right=-0.105cm of a1]   {};   
					\node [port] (a3) [right=-0.805cm of a1]  {$in_p$};

					\node[bubble] (a4) [left=-0.105cm of a1]   {};   
					\node [port] (a5) [left=-0.99cm of a1]  {$out_p$};  
					
					\node[type]  [above=-0.6cm of a1]{{\small Pipe $B(1,1)$}};

					\node [component] (b2) [right=2cm of a1]  {};
					\node[bubble] (b3) [right=-0.105cm of b2]   {};   
					\node [port] (b4) [right=-0.805cm  of b2]  {$in_f$};

					\node[bubble] (b5) [left=-0.105cm of b2]   {};   
					\node [port] (b6) [left=-0.99cm of b2]  {$out_f$};  
					
					\node[type]  [above=-0.6cm of b2]{{\small Filter $B(2,1)$}};
					\path[-]          (a3)  edge                  node {} (b6);

					\node [component] (c2) [right=2cm of b2]{};
					
					\node[bubble] (c3) [right=-0.105cm of c2]   {};   
					\node [port] (c4) [right=-0.805cm  of c2]  {$in_p$};

					\node[bubble] (c5) [left=-0.105cm of c2]   {};   
					\node [port] (c6) [left=-0.99cm of c2]  {$out_p$};  
					
					\node[type]  [above=-0.6cm of c2]{{\small Pipe $B(1,2)$}};
					
					\path[-]          (b4)  edge                  node {} (c6);

					\node [component] (d2) [above right= 0.5cm and 2cm of c2]  {};
					\node[bubble] (d3) [right=-0.105cm of d2]   {};  
					\node [port] (d4) [right=-0.805cm  of d2]  {$in_f$};

					\node[bubble] (d5) [left=-0.105cm of d2]   {};   
					\node [port] (d6) [left=-0.99cm of d2]  {$out_f$};  
					\node[] (i1) [above left=-0.3 cm and -0.20cm of d6]   {};
					
					\node[type]  [above=-0.6cm of d2]{Filter $B(2,2)$};
					\path[-]          (c3)  edge                  node {} (i1);

					\node [component] (e2) [below right= 0.5cm and 2cm of c2]  {};
					\node[bubble] (e3) [right=-0.105cm of e2]   {};   
					\node [port] (e4) [right=-0.805cm  of e2]  {$in_f$};

					\node[bubble] (e5) [left=-0.105cm of e2]   {};   
					\node [port] (e6) [left=-0.99cm  of e2]  {$out_f$}; 
					\node[] (i2) [above left=-0.62 cm and -0.30cm of e6]   {};

					\node[type]  [above=-0.6cm of e2]{{\small Filter $B(2,3)$}};
					\path[-]          (c3)  edge                  node {} (i2);

					\node [component] (f2) [right=2cm of d2]{};
					
					\node[bubble] (f3) [right=-0.105cm of f2]   {};   
					\node [port] (f4) [right=-0.805cm  of f2]  {$in_p$};

					\node[bubble] (f5) [left=-0.105cm of f2]   {};   
					\node [port] (f6) [left=-0.99cm  of f2]  {$out_p$};  
					\node[] (i3) [above left=-0.43 cm and -0.25cm of f6]   {};

					\node[type]  [above=-0.6cm of f2]{{\small Pipe $B(1,3)$}};
					
					\path[-]          (d3)  edge                  node {} (i3);
					
					\node [component] (g2) [right=2cm of e2]{};
					
					\node[bubble] (g3) [right=-0.105cm of g2]   {};   
					\node [port] (g4) [right=-0.805cm  of g2]  {$in_p$};

					\node[bubble] (g5) [left=-0.105cm of g2]   {};   
					\node [port] (g6) [left=-0.99cm  of g2]  {$out_p$};  
					\node[] (i4) [above left=-0.43 cm and -0.25cm of g6]   {};

					\node[type]  [above=-0.6cm of g2]{{\small Pipe $B(1,4)$}};
					
					\path[-]          (e3)  edge                  node {} (i4);
					
			\end{tikzpicture}}
			\caption{ Pipes-Filters architecture}
			\label{p-f}
		\end{figure}
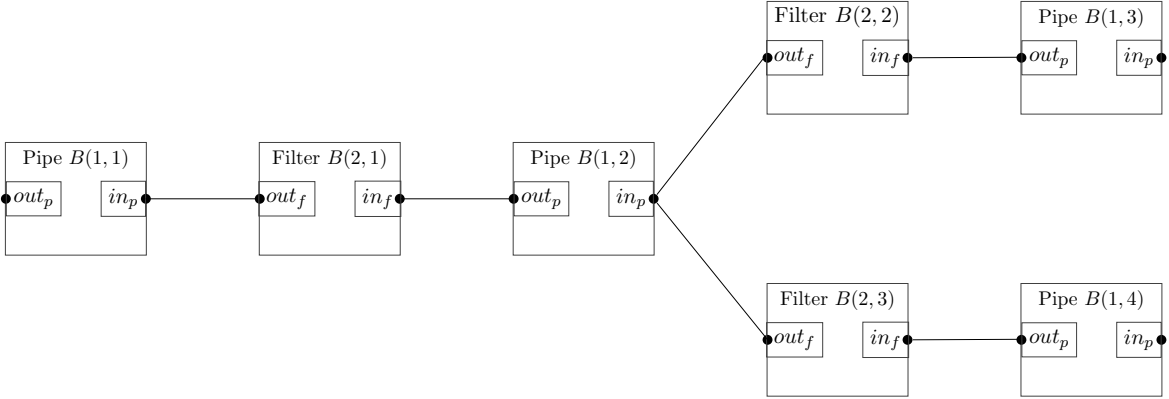 
	    \noindent Next, we let 
	      $$\zeta_2(out_p(r))= \bigsqcup_{j_1 \in [u_2]}\left( \bigwedge_{j_2 \in [u_2] \setminus \{j_1\}} \neg{\phi_{in_fout_p}(in_f(j_2),out_p(r))}\right).$$\\
	    If $r\in[u_1]$. Then this formula describes that the output of the pipe instance  $B(1,r) (r\in[u_1])$, can be connected to at most one a filter instance $B(2,j_1)$($j_1\in[u_2]$).
	    Moreovoer, we define the formulas
    
    $$\zeta_3\left(k_1\right) =\bigsqcup_{i_1\in [u_2]}  \left( \phi_{in_fout_p}(in_f(i_1),out_p(k_1))\right) \wedge \\ \bigwedge_{i\in [u_2]} \left( \neg {\phi_{out_fin_p}(out_f(i),in_p(k_1))}\right).$$
    
   $$\zeta_4\left(k_2\right) =\bigsqcup_{i_2\in [u_2]}  \left( \phi_{out_fin_p}(out_f(i_2),in_p(k_2))\right) \wedge \\ \bigwedge_{i\in [u_2]}\left( \neg {\phi_{in_fout_p}(in_f(i),out_p(k_2))}\right).$$\\
  Let $k_1,k_2\in[u_1]$. The formula $\zeta_3\left(k_1\right)$ (resp. $\zeta_4\left(k_2\right)$) interpret that the output (resp. input) of the pipe $B(1,k_1)$ (resp. $B(1,k_2)$) interacts with the the input (resp. output) of a filter but the input (resp. output) of this pipe is disconnected.\\
  We conclude to the \emph{PCL} formula 
	$$
	 f = \sum_{j \in [u_2]}\zeta_1(in_f(j),out_f(j)) \wedge \\ \bigwedge_{r \in [u_1]}\zeta_2\left( out_p(r)\right) \wedge  \bigsqcup_{k_1\in[u_1]\atop k_2\in[u_1]\setminus\{k_1\}}\left( \zeta_3\left(k_1\right)\wedge \zeta_4\left(k_2\right)\right).
    $$\\
	If $\gamma \in C_{p\B}$ satisfies $f$, then it is clear that every filter instance $B(2,j)$ ($j \in [u_2]$)  is connected with two pipe instances $B(1,i_1)$, $B(1,i_2)$ ($i_1,i_2\in[u_1]$), so the architecture is active. An instantiation of system with four pipes and three filters at a time $g\geq0,$ where the set of the active interactions follow the pipe-filter architecture, is given in Figure \ref{p-f}.
		
	\end{example}	
	
	\begin{example}
		Next we will develop a \emph{PCL} formula for the \emph{Star architecture}. This architecture contains only one component type $B(1)$ with a unique port $p$. There is an instance considered as the \emph{central component}. The rest instances can  be connected to the central one, with the restriction that at least one of them is always connected to it, so that the architecture is active. No other interactions are permitted.
	
		\begin{figure}[h]
			\centering
			\resizebox{0.6\linewidth}{!}{
				\begin{tikzpicture}[>=stealth',shorten >=1pt,auto,node distance=1cm,baseline=(current bounding box.north)]
					\tikzstyle{component}=[rectangle,ultra thin,draw=black!75,align=center,inner sep=9pt,minimum size=1.8cm]
					\tikzstyle{port}=[rectangle,ultra thin,draw=black!75,minimum size=6mm]
					\tikzstyle{bubble} = [fill,shape=circle,minimum size=5pt,inner sep=0pt]
					\tikzstyle{type} = [draw=none,fill=none] 
					
					\node [component] (a1) {};
					\node [port] (a2) [above=-0.605cm of a1]  {$p$};
					\node[bubble] (a3) [above=-0.105cm of a1]   {};

					\node []  (s1)  [below=0.75cm of a3]  {$B(1,1)$}; 
					
					\node [component] (a4) [above left =2.cm and 4cm of a1]  {};
					\node [port] (a5) [below=-0.605cm of a4]  {$p$};
					\node[] (i1) [above left=-0.15 cm and -0.73cm of a4]   {};
					\node[bubble] (a6) [below=-0.105cm of a4]   {};

					\node []  (s2)  [above=0.75cm of a6]  {$B(1,2)$}; 
					
					\path[-]          (a3)  edge                  node {} (a6);
					\path[-]          (a6)  edge                  node {} (a3);
					
					\node [component] (b4) [right =2 cm of a4]  {};
					\node [port] (b5) [below=-0.605cm of b4]  {$p$};
					\node[] (i2) [above left=-0.25 cm and -0.30cm of b5]   {};
					\node[bubble] (b6) [below=-0.105cm of b4]   {};
					
					\node []  (s3)  [above=0.75cm of b6]  {$B(1,3)$};

					\path[-]          (a3)  edge                  node {} (b6);
					\path[-]          (b6)  edge                  node {} (a3);
					
					\node [component] (c4) [right =2 cm of b4]  {};
					\node [port] (c5) [below=-0.605cm of c4]  {$p$};
					\node[] (i3) [above right=-0.20 cm and -0.38cm of c5]   {};
					\node[bubble] (c6) [below=-0.105cm of c4]   {};
					
					\node []  (s4)  [above=0.75cm of c6]  {$B(1,4)$};

					\node [component] (d4) [right =2 cm of c4]  {};
					\node [port] (d5) [below=-0.605cm of d4]  {$p$};
					\node[] (i4) [above right=-0.20 cm and -0.27cm of d5]   {};
					\node[bubble] (d6) [below=-0.105cm of d4]   {};
					
					\node []  (s5)  [above=0.75cm of d6]  {$B(1,5)$}; 
					\path[-]          (a3)  edge                  node {} (d6);
					\path[-]          (d6)  edge                  node {} (a3);

			\end{tikzpicture}}
			\caption{Star architecture.}
			\label{star}
		\end{figure}
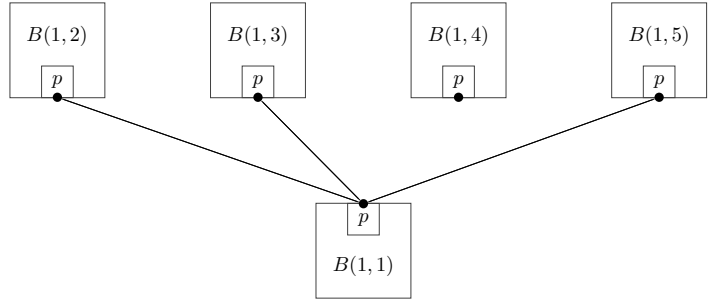

\noindent		The \emph{PIL} formula describing the interaction among the instances $B(1,j)$ and $B(1,i)$ ($j,i\in[u_1]$) is given by 
$$\phi_{pp}(p(j),p(i))=m_{\{p(j),p(i)\}}.$$\\
Then,		the \emph{PCL} formula for Star architecture is as follows
		$$f=\bigsqcup_{j\in[u_1]}\bigvee_{i\in[u_1]\setminus\{j\}}\phi_{pp}(p(j),p(i)).$$		
If the  unique central component $B(1,j)$ for some $j\in[u_1]$ is connected to a separate instance $B(1,i)$ ($i\in[u_1]$),	 then the corresponding  configuration $\gamma \in C_{p\B}$ satisfies the \emph{PCL} formula $f$.
	\end{example}

	\begin{example}

		\definecolor{harlequin}{rgb}{0.25, 1.0, 0.0}
		\definecolor{ao}{rgb}{0.0, 0.0, 1.0}
		
		\definecolor{darkorange}{rgb}{1.0, 0.55, 0.0}
		\definecolor{harlequin}{rgb}{0.25, 1.0, 0.0}
		\definecolor{ao}{rgb}{0.0, 0.0, 1.0}

		\begin{figure}[h]
			\centering
			\resizebox{0.65\linewidth}{!}{
				\begin{tikzpicture}[>=stealth',shorten >=1pt,auto,node distance=1cm,baseline=(current bounding box.north)]
					\tikzstyle{component}=[rectangle,ultra thin,draw=black!75,align=center,inner sep=9pt,minimum size=1.5cm,minimum height=3.4
					cm,minimum width=2.6cm]
					\tikzstyle{port}=[rectangle,ultra thin,draw=black!75,minimum size=6mm]
					\tikzstyle{bubble} = [fill,shape=circle,minimum size=5pt,inner sep=0pt]
					\tikzstyle{type} = [draw=none,fill=none] 
					
					\node [component,align=center] (a1)  {};
					\node [port] (a2) [above right= -1.6cm and -0.62cm of a1]  {$b_1$};
					\node[bubble] (a3) [above right=-1.35cm and -0.08cm of a1]   {};

					\node [port] (a4) [above right= -2.8cm and -0.63cm of a1]  {$b_2$};
					\node[bubble] (a5) [above right=-2.57cm and -0.08cm of a1]   {};

					\node[type] (a6) [above right=-1.5cm and -2.4cm of a1]{Blackb.};
					\node            [below=-0.1cm of a6]{$B(1,1)$};

					\tikzstyle{control}=[rectangle,ultra thin,draw=black!75,align=center,inner sep=9pt,minimum size=1.5cm,minimum height=2.3cm,minimum width=3.5cm]
					
					\node [control,align=center] (b1)  [below right=8cm and 2.5cm of a1]{};
					\node [port] (b2) [above left=-0.605cm  and -1.65 of b1]  {$c_1$};
					\node[bubble] (b3) [above left=-0.1cm and -0.35cm of b2]   {};

					\node [port] (b6) [right=0.4 of b2]  {$c_2$};
					\node[bubble] (b7) [above left=-0.1cm and -0.35cm of b6]   {};

					\node[type] (b8) [below=-1.2cm of b1]{Contr.};
					\node[]      [below=-0.1cm of b8]{$B(3,1)$};
					
					\tikzstyle{source}=[rectangle,ultra thin,draw=black!75,align=center,inner sep=9pt,minimum size=1.5cm,minimum height=3.5cm,minimum width=3cm]
					
					\tikzstyle{port}=[rectangle,ultra thin,draw=black!75,minimum size=6mm,inner sep=3pt]
					
					\node [source,align=center] (c1) [above right=-1cm and 9cm of a1] {};
					\node [port] (c2) [above left= -1.7cm and -0.6cm of c1]  {$s_1$};
					\node[bubble] (c3) [above left=-0.35cm and -0.08cm of c2]   {};

					\tikzstyle{port}=[rectangle,ultra thin,draw=black!75,minimum size=6mm,inner sep=2.7pt]
					\node [port] (c6) [above left= -1.6cm and -0.6cm of c2]  {{\small $s_2$}};
					\node[bubble] (c7) [above left=-0.35cm and -0.08cm of c6]   {};

					\node[type,align=center] (c8) [above right=-1.5cm and -1.6cm of c1]{Sour.};
					\node[] [below=-0.1cm of c8]{$B(2,1)$};
					
					\tikzstyle{port}=[rectangle,ultra thin,draw=black!75,minimum size=6mm,inner sep=3pt]
					\node [source,align=center] (d1) [below= 1cm of c1] {};
					\node [port] (d2) [above left= -1.7 cm and -0.6cm of d1]  {$s_1$};
					\node[bubble] (d3) [above left=-0.35cm and -0.08cm of d2]   {};

					\tikzstyle{port}=[rectangle,ultra thin,draw=black!75,minimum size=6mm,inner sep=2.7pt]
					\node [port] (d6) [above left= -1.6cm and -0.6cm of d2]  {{\small $s_2$}};
					\node[bubble] (d7) [above left=-0.35cm and -0.08cm of d6]   {};

					\node[type] (d8) [above right=-1.5cm and -1.6cm of d1]{Sour.};
					\node[] [below=-0.1cm of d8]{$B(2,2)$};
					
					\tikzstyle{port}=[rectangle,ultra thin,draw=black!75,minimum size=6mm,inner sep=3pt]
					\node [source,align=center] (e1) [below= 1cm of d1] {};
					\node [port] (e2) [above left= -1.7 cm and -0.6cm  of e1]  {$s_1$};
					\node[bubble] (e3) [above left=-0.35cm and -0.08cm of e2]   {};

					\tikzstyle{port}=[rectangle,ultra thin,draw=black!75,minimum size=6mm,inner sep=2.7pt]
					\node [port] (e6) [above left= -1.6cm and -0.6cm of e2]  {{\small $s_2$}};
					\node[bubble] (e7) [above left=-0.35cm and -0.08cm of e6]   {};
					
					\
					
					\node[type] (e8) [above right=-1.5cm and -1.6cm of e1]{Sour.};
					\node[] [below=-0.1cm of e8]{$B(2,3)$};
					
					
					\path[-]          (a3)  edge                  node {} (c3);
					\path[-]          (a3)  edge                  node {} (d3);
					\path[-]          (a3)  edge                  node {} (e3);

					\path[-]          (a3)  edge                  node {} (b3);
					\path[-]          (b3)  edge                  node {} (a3);
					
					
					\path[-]          (c3)  edge                  node {} (a3);
					\path[-]          (d3)  edge                  node {} (a3);
					\path[-]          (e3)  edge                  node {} (a3);

					
					\path[-]          (a5)  edge  [harlequin]                node {} (d7);
					\path[-]          (a5)  edge  [harlequin]                node {} (e7);

					\path[-]          (a5)  edge  [harlequin]                node {} (b7);
					
					
					\path[-]          (d7)  edge   [harlequin]               node {} (a5);
					\path[-]          (e7)  edge   [harlequin]               node {} (a5);

					\path[-]          (b7)  edge   [harlequin]              node {} (a5);

					
					\path[-]          (b7)  edge  [harlequin]               node {} (d7);
					\path[-]          (d7)  edge  [harlequin]               node {} (b7);
					
					\path[-]          (b7)  edge  [harlequin]               node {} (e7);
					\path[-]          (e7)  edge  [harlequin]               node {} (b7);
					
					,	\end{tikzpicture}}
			\caption{Blackboard architecture.}
			\label{blackboard}
		\end{figure}
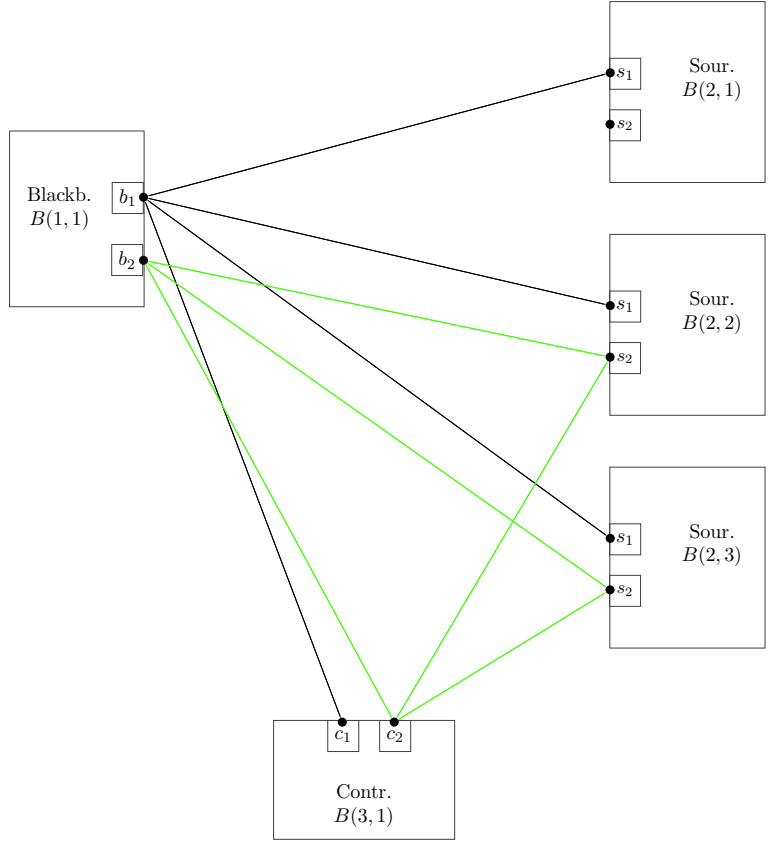
		We deal with the \emph{Blackboard} architecture \cite{Bu:Pa,Co:Bl,Me:We} which contributes to important applications (cf. for instance \cite{Ni:Bl,St:Cr}). It involves  a central component called \emph{Blackboard} which is denoted by $B(1)$,  and presents  the state of a problem to be solved, the separate independent \emph{Knowledge Sources} components ($B(2)$) representing the experts 
		providing solutions for the given problem, and  the \emph{Controller} component which is denoted by $B(3)$ and schedules the candidate knowledge sources to execute their solutions on the blackboard. Every component has two ports. Specifically, blackboard has the ports $b_1, b_2$,  knowledge sources have the ports $s_1,s_2$, and controller has the ports $c_1$ and $c_2$.\\
Blackboard interacts with knowledge sources and controller, respectively via $\{b_1,s_1\}$ and ${\{b_1,c_1\}}$  in order to present the state of the problem to sources and inform the controller. The connection of a source to blackboard, via the controller, in order to write its solution to blackboard is achieved by an interaction of type $\{b_2,s_2,c_2\}$.\\
The \emph{PIL} formula 
$$\phi_{b_1s_1}(b_1(1),s_1(j))=m_{\{b_1(1),s_1(j)\}}$$
describes the connection of blackboard $B(1,1)$ to knowledge source $B(2,j)$ ($j \in [u_2]$), so that the latter is being informed by blackboard for a given problem. 
The connection among blackboard $B(1,1)$ and controller $B(3,1)$ is given by
		$$\phi_{b_1c_1}(b_1(1),c_1(1))=m_{\{b_1(1),c_1(1)\}}.$$	\\
Finally, the \emph{PIL} formula describing the interaction among the blackboard $B(1,1)$, a knowledge source $B(2,j)$ ($j \in [u_2]$) and the controller $B(3,1)$ is given by
		$$\phi_{b_2s_2c_2}(b_2(1),s_2(j),c_2(1))=m_{\{b_2(1),s_2(j),c_2(1)\}}.$$ \\
		The \emph{PCL} formula describing the Blackboard architecture is
		$$f=\phi_{b_1c_1}(b_1(1),c_1(1))+\sum_{j\in[u_2]}\phi_{b_1s_1}(b_1(1),s_1(j))+\bigvee_{j\in [u_2]}\phi_{b_2s_2c_2}(b_2(1),s_2(j),c_2(1)).$$\\
		Let $\gamma \in C_{p\B}$ satisfies $f$. Then, the blackboard $B(1,1$) informs all knowledge sources and the controller  for a given problem, and at least one   knowledge source  is connected to blackboard and controller in order to write its solution (cf. Figure \ref{blackboard} for a  snapshot of the dynamic reconfigurable system with three knowledge sources where all the active interactions satisfy formula $f$). \\
	
	\end{example}

	\begin{example}
We continue with \emph{Request-Response} architecture considered as dynamic reconfigurable system. This architecture is very popular in distributed computing and is suitable for a wide variety of applications. Request-Response model has three  component types namely \emph{Services}, \emph{Clients}, and \emph{Coordinators}. We will denote them by $B(1), B(2)$ and $B(3)$, respectively. Services contain two ports  $get_s$ and $send$. Clients and coordinators have both three ports, specifically  $con_{cl},req,rec$, and  $con_c,get_c,dsc$, respectively. A service communicates with a client through a coordinator. More precisely, a client issues a request to a service via an interaction of type $\{get_s,req,get_c\}$. The service process the request and returns a response via an interaction of type $\{send,rec,dsc\}$. Also, the client is connected to a coordinator via  an interaction of type $\{con_{cl},con_c\}$. The coordinator controls that only one client is connected to a service. The service instance can not be connected to any other client and the coordinator can not be connected to any other service.\\
The  interaction among a client $B(2,k)$($k\in [u_2]$) and a coordinator  $B(3,j)$($j\in [u_3]$) is given by the \emph{PIL} formula $$\phi_{cc}(con_{cl}(k),con_c(j))=m_{\{con_{cl}(k),con_c(j)\}}.$$
The interaction of type $\{get_s,req,get_c\} $ among a service $B(1,i)$($i\in [u_1]$), a client $B(2,k)$($k\in [u_2]$), and a coordinator $B(3,j)$($j\in [u_3]$)  is given by $$\phi_{grg}(get_s(i),req(k),get_c(j))=m_{\{get_s(i),req(k),get_c(j)\}},$$
and the interaction of type $\{send,rec,dsc\}$  is given by the \emph{PIL} formula 	$$\phi_{srd}({{send(i),rec(k),dsc(j)}})=m_{\{send(i),rec(k),dsc(j)\}}.$$
We consider the formula
\begin{multline*}
			\zeta(get_s(i),send(i))=\bigsqcup_{k\in[u_2] \atop j\in [u_3]} \Bigl(\phi_{cc}(con_{cl}(k),con_c(j))+\phi_{grg}({get_s(i),req(k),get_c(j)})+\\ \phi_{srd}(send(i),rec(k),dsc(j))\Bigr)
\end{multline*}
which describes the connection of the service instance $B(1,i)$  with a single client  $B(2,k)$ ($k\in [u_2]$) through one coordinator instance $B(3,j)$ ($j\in[u_3]$). We let
		\begin{multline*}
			\xi=\bigwedge_{i_1 \in [u_1], k_1\in [u_2] \atop j_1\in [u_3]} \left(   \neg{\phi_{grg}\left(  get_s(i_1),req(k_1),get_c(j_1)\right)  } \sqcup \left(\sim  \phi_{grg}(get_s(i_1),req(k_1),get_c(j_1)) \wedge \xi_1  \right) \right) 
	 \end{multline*}
with
    $$\xi_1=\bigwedge_{i_2\in[u_1] \setminus\{i_1\} \atop k_2\in[u_2]}\neg{\phi_{grg}(get_s(i_2),req(k_2),get_c(j_1))} .$$
Formula $\xi$ ensures that if there is a connection between the three instances $B(1,i_1),B(2,k_1)$ and $B(3,j_1)$, then the coordinator instance $B(3,j_1)$ cannot interact with any other service instance.\\
Finally, we conclude to the \emph{PCL} formula
		$$f=\bigvee_{i\in [u_1]} \big( \zeta(get_s(i),send(i)) \wedge  \xi\big).$$ \\
		
					\definecolor{db}{RGB}{000,000,000}
		\definecolor{fb}{RGB}{255,255,255}
		\begin{figure}[h] 
			\begin{center}
				\begin{tikzpicture}[scale=0.9]
					\foreach \Z in {0,0.3,0.6}
					{\draw[fill=fb,draw=db] (-2,-2,\Z) rectangle (2,1,\Z);}
					\node at (-1.00,0) { \tiny Coordinator};
					\node at (-1,-0.6) { \tiny $B(3,j)$};
					\draw  (-1.5,-2.23) rectangle (-0.8,-1.65);
					\node at (-1.15,-1.94) {\tiny $dsc$};
					\draw[fill] (-1.15,-2.23) circle [radius=1.5pt];
					\draw  (0.4,-2.23) rectangle (1.1,-1.65);
					\node at (0.75,-1.94) {\tiny $get_c$};
					\draw[fill] (0.75,-2.23) circle [radius=1.5pt];
					\draw  (1.17,-0.3) rectangle (1.77,-1);
					\node at (1.47,-0.65) {\tiny $con_c$};
					\draw[fill] (1.77,-0.65) circle [radius=1.5pt];

					\foreach \Z in  {0,0.3,0.6}
					{\draw[fill=fb,draw=db] (-5,-7,\Z) rectangle (-1.3,-5,\Z);}
					\node at (-4.2,-6.3) { \tiny Service};
					\node at (-4.2,-6.7) {\tiny $B(1,i)$};
					\draw[fill=fb,draw=db]  (-4.4,-5.87) rectangle (-3.7,-5.22);
					\node at (-4.0,-5.54) {\tiny $get_s$};
					\draw[fill] (-4.0,-5.22) circle [radius=1.5pt];
					\draw  (-3.0,-5.87) rectangle (-2.3,-5.22);
					\node at (-2.65,-5.545) {\tiny $send$};
					\draw[fill] (-2.65,-5.22) circle [radius=1.5pt];
					
					\foreach \Z in  {0,0.3,0.6,0.9}
					{\draw[fill=fb,draw=db] (1,-7,\Z) rectangle (4.7,-5,\Z);}
					\node at (1.6,-6.4) {\tiny Client};
					\node at (1.6,-6.8) {\tiny $B(2,k)$};
					\draw  (1.1,-5.97) rectangle (1.8,-5.345);
					\node at (1.45,-5.66) {\tiny $req$};
					\draw[fill] (1.45,-5.36) circle [radius=1.5pt];
					\draw  (2.1,-5.97) rectangle (2.8,-5.345);
					\node at (2.45,-5.645) {\tiny $rec$};
					\draw[fill] (2.45,-5.36) circle [radius=1.5pt];
					\draw  (3.1,-5.97) rectangle (3.8,-5.345);
					\node at (3.45,-5.645) {\tiny $con_s$};
					\draw[fill] (3.45,-5.36) circle [radius=1.5pt];
					
					\draw (-4.0,-5.22)--(-4.0,-4.22);
					\draw (-4.0,-4.22)--(1.45,-4.22);
					\draw (1.45,-4.22)--(1.45,-5.36);
					\draw (0.75,-4.22)--(0.75,-2.23);
					
					\draw (-2.65,-5.22)--(-2.65,-3.72);
					\draw (2.45,-3.72)--(-2.65,-3.72);
					\draw (2.45,-3.72)--(2.45,-5.36);
					\draw (-1.15,-3.72)--(-1.15,-2.23);
					
					\draw  (1.77,-0.65)--(3.45,-0.65);
					\draw  (3.45,-0.65)--(3.45,-5.36);
					
				\end{tikzpicture}  
				
			\end{center}
			\caption{ Request-Response architecture.}
			\label{r-r}
		\end{figure}
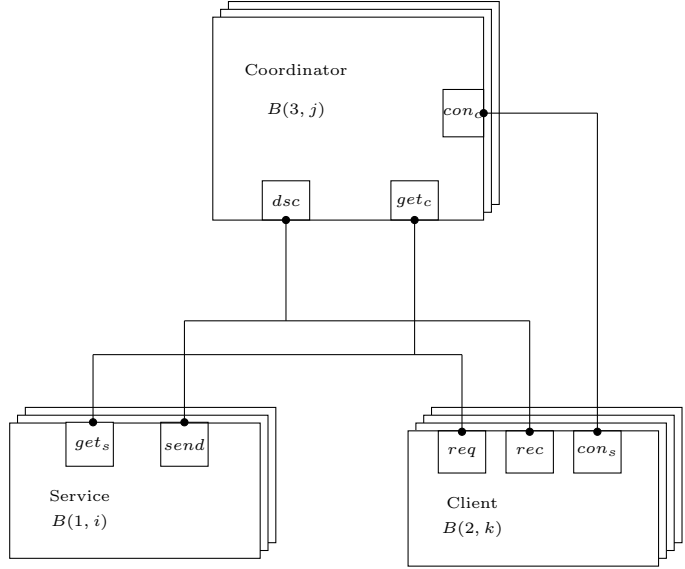
\noindent Trivially, if a configuration $\gamma \in C_{p\B}$ satisfies  $f$, then there is at least one service instance $B(1,i)$ ($i\in [u_1]$)  connected to a unique client  $B(2,k)$ ($k\in [u_2]$) through one coordinator instance $B(3,j)$ ($j\in[u_3]$) which implies that the architecture is active. 
	\end{example}
		
	\subsection{Decidability results for dra}
	In the sequel, we investigate decidability properties for dynamic reconfigurable architectures. First, we show that for every architecture $(p\B, f )$ we can compute the set of trustworthy implementations $\sigma$ for $f$.

\begin{theorem}\label{compute_trust}
		Let $(p\B, f)$ be an architecture. Then we can compute, in doubly exponential time, the set of implementations on $p\B$ which are trustworthy for $f$.
\end{theorem}
\begin{proof}
Firstly, we compute the set $C_{p\B}$ in doubly exponential time. Next, by  Theorem 4.43. in \cite{Ma:Co}, we can effectively construct a PCL formula $f'$ in full normal form which is equivalent to $f$. For this, we need also a doubly exponential time. Using the structure of $f'$, we determine in constant time, the set $C_{p\B}( \models f') \subseteq C_{p\B}$ with  
$$C_{p\B}( \models f')= \{\gamma \in C_{p\B} \mid \gamma \models f'\}.$$

\noindent Next, we consider the B\"{u}chi automaton $\mathcal{A}_{f'}=\left(\{q_0,q_1\},C_{p\mathcal{B}},\{q_0\},\Delta,\{q_1\} \right) $  with $$\Delta=\{\left(q_0,\gamma,q_0\right)\mid\gamma\in C_{p\mathcal{B}} \} \cup \{\left(q_0,\gamma,q_1\right)\mid\gamma\in C_{p\B}(\models f')\}\cup \{\left(q_1,\gamma,q_1\right)\mid\gamma\in C_{p\B}(\models f')\}.$$		
		By construction of $\mathcal{A}_{f'}$, we get that a sequence $\left(\gamma_i\right)_{i\geq 0} \in C_{p\mathcal{B}}^\omega $ is accepted by $\mathcal{A}_{f'}$ iff there is a natural number $g\geq0$ such that $\gamma_i \models f'$ for every $i>g$, i.e., $\gamma_i \models f$ for every $i>g$.  We conclude our result by Remark \ref{rem_imp}.
\end{proof}	

\begin{corollary} \label{col_trust}
Let $(p\B, f, \sigma)$ be a dra. Then we can decide, in doubly exponential time, whether $\sigma$ is trustworthy for $f$ or not.
\end{corollary}	
\begin{proof}
By Theorem \ref{compute_trust} we compute, in doubly exponential time, the set of implementations on $p\B$ which are trustworthy for $f$. Hence, we only need to check whether $\sigma$ belongs to that set. 
\end{proof}	

\

Next we define a type of equivalence for dra's. More precisely, let $(p\B, f, \sigma)$ and $(p\B', f', \sigma')$ be two dra's. Then $(p\B, f, \sigma)$ and $(p\B', f', \sigma')$ are called \emph{equivalent} if $p\B=p\B'$, $f \equiv f'$, and $\sigma$, $\sigma'$ are both trustworthy for $f$ and $f'$, respectively. We show the next proposition.  		
\begin{proposition}		   
Let $(p\B, f, \sigma)$ and $(p\B', f', \sigma')$ be two dra's. We can decide, in doubly exponential time, whether $(p\B, f, \sigma)$ and $(p\B', f', \sigma')$  are equivalent  or not.
\end{proposition}	
\begin{proof}
The equality $p\B=p\B'$ is decidable in linear time, and the equivalence $f\equiv f'$ is decidable in doubly exponential time (cf. Theorem in \cite{Pa:We}). Finally,  the trustworthiness of the implementations $\sigma$ and $\sigma'$ can be decided, by Corollary \ref{col_trust}, in doubly exponential time, hence we conclude our result.
\end{proof}

\

Finally we introduce the notion of partially trustworthiness for implementations. Specifically, let $(p\B,f, \sigma)$ be a dra. Then, the  implementation $\sigma$ is called  \emph{partially trustworthy for} $f$ if there are $\gamma'_l \subseteq \gamma(\sigma_l)$ for every $l \geq 0 $ and a natural number $g \geq0$  such that $\gamma'_l \models f$   for every $l \geq g$. Trivially, every trustworthy implementation for $f$ it is also partially trustworthy for $f$. Though, as we show in the sequel, we can decide whether a non-trustworthy implementation is partially trustworthy. For this we need some preliminary matter. \\
Let $p\B$ be a component-based system and $f=\bigsqcup\nolimits_{i\in I}\sum\nolimits_{i\in J_{i}}m_{i,j}$ a PCL formula over $P_{p\B}$ in full normal form. Trivially every full monomial $m_{i,j}$ determines a unique interaction $a_{i,j}$ for every $i \in I$ and $j \in J_i$. Next let $\gamma \in C_{p\B}$ such that $\gamma \not \models f$. Hence $\gamma \not \models \sum\nolimits_{j\in J_{i}}m_{i,j}$ for every $i \in I$ which in turn implies that $\gamma \neq \{a_{i,j} \in I_{p\B} \mid j \in J_i\}$ for every $i \in I$. We set $\mathcal{D}_f=\{\gamma \in C_{p\B} \mid \{a_{i,j} \mid j\in J_{i}\} \subseteq \gamma \text{ for some } i \in I  \}$.

\begin{proposition}\label{part-trust}
Let	$(p\B, f, \sigma)$ be a dra such that $\sigma$ is not trustworthy for $f$. Then, we can decide whether $\sigma$ is partially trustworthy for $f$ or not.
\end{proposition} 	
\begin{proof}
Taking into account the proof of Theorem \ref{compute_trust} we compute the set $\mathcal{D}_{f'}$. Since we have already computed $C_{p\B}$ and $C_{p\B}(\models f')$ this requires a constant time. We construct the B\"{u}chi automaton $\mathcal{B}_{f'}=\left(\{q'_0,q'_1\},C_{p\mathcal{B}},\{q'_0\},\Delta,\{q'_1\} \right) $  with $$\Delta=\{\left(q'_0,\gamma,q'_0\right)\mid\gamma\in C_{p\mathcal{B}} \} \cup \{\left(q'_0,\gamma,q'_1\right)\mid\gamma\in   \mathcal{D}_{f'}\}\cup \{\left(q'_1,\gamma,q'_1\right)\mid\gamma\in  \mathcal{D}_{f'} \}.$$
We conclude our result by  a similar argument as in Theorem \ref{compute_trust} and Corollary \ref{col_trust}.
\end{proof}

\

We should not that the decidability of the partial trustworthiness of $\sigma$ does not require any additional time, since by Theorem \ref{compute_trust} we have already shown that $\sigma$ is not trustworthy in doubly exponential time. Furthermore, our last result gives us the possibility to "correct" a partially trustworthy implementation. Indeed, keeping the above notations we consider the implementation $\sigma'=\left(\sigma_l|_{\gamma'_l}\right)_{l \geq 0}$ which is defined for every $l \geq 0$ and $a \in I_{p\B}$, by $\sigma_l|_{\gamma'_l}(a)=\mathrm{active}$ if $a \in \gamma'_l$ and $\sigma_l|_{\gamma'_l}(a)=\mathrm{inactive}$ otherwise. Trivially, the implementation $\sigma'$ is now trustworthy for $f$.  

\

\noindent \textbf{Future work.}

\noindent i) The complexity of our "correction" method has to be computed.
 
\noindent ii) An implementation $\sigma$ may be non-trustworthy since: (a) it is partially trustworthy (see above), (b) there are infinitely many l's such that a required interaction is missing in $\gamma(\sigma_l)$, or (c) an infinite combination of both the cases (a) and (b). It is a challenge to determine whether the non-trustworthiness of an implementation is due to case (b) or (c) and even to find a way to "correct" it. 

\noindent iii) It should be interesting to develop a software tool implementing our theory.

\noindent iv) We believe that it is important to study model-checking techniques for reconfigurable architectures described by PCL. In \cite{Ho:Te} PCL is evolved but its formulas participate as atomic propositions in an LTL.

\end{document}